\documentclass{lmcs}
\pdfoutput=1

\usepackage{lastpage}
\lmcsdoi{18}{3}{25}
\lmcsheading{}{\pageref{LastPage}}{}{}%
{Sep.~17,~2021}{Aug.~22,~2022}{}

\usepackage[utf8]{inputenc}

\newcommand{\A}{\mathbf{A}}
\newcommand{\B}{\mathbf{B}}
\newcommand{\C}{\mathbf{C}}
\newcommand{\G}{\mathbf{G}}
\newcommand{\D}{\mathbf{D}}
\newcommand{\N}{\mathbb{N}}
\newcommand{\K}{\mathbf{K}}
\newcommand{\X}{\mathbf{X}}
\newcommand{\Z}{\mathbb{Z}}

\newcommand{\CSP}{\mathrm{CSP}}
\newcommand{\PCSP}{\mathrm{PCSP}}

\newcommand{\compP}{{\sf{P}}}
\newcommand{\compNP}{{\sf{NP}}}

\newcommand{\set}[1]{[#1]}
\newcommand{\st}{\;:\;}

\newlength{\ldprobleft} 
\newlength{\ldprobmid}  
\newlength{\ldprobright}

\newcommand{\dproblem}[3]{
\begin{equation*}
\parbox{\textwidth}{
\begin{tabular}{ @{} p{\ldprobleft} p{\ldprobmid} p{\ldprobright} @{} }
& \multicolumn{2}{l}{\textbf{#1}} \\
& {\begin{minipage}[t]{\ldprobmid}Input:\vspace{1.5pt}\end{minipage}} & {\begin{minipage}[t]{\ldprobright}#2\vspace{1.5pt}\end{minipage}} \\ 
& {\begin{minipage}[t]{\ldprobmid}Output:\end{minipage}} & {\begin{minipage}[t]{\ldprobright}#3\end{minipage}} \\
\end{tabular}}
\end{equation*}
}

\keywords{Promise Constraint Satisfaction Problem, Constraint Satisfaction Problem, relational structure homomorphism, digraph, sandwich}

\begin{document}

\title{Small Promise CSPs that reduce to large CSPs}
\author{Alexandr Kazda\lmcsorcid{0000-0002-7338-037X}}[a]
\author{Peter Mayr\lmcsorcid{0000-0003-1163-313X}}[b]
\author{Dmitriy Zhuk\lmcsorcid{0000-0003-0047-5076}}[c]
\address{Department of Algebra, Faculty of Mathematics and Physics, Charles University, Prague, Czech Republic}
\email{alex.kazda@gmail.com}

\address{Department of Mathematics,
University of Colorado,
Boulder, Colorado, USA}
\email{mayr@colorado.edu}

\address{
HSE University, Russia and 
Lomonosov Moscow State University, Russia
}
\email{zhuk@intsys.msu.ru}

\begin{abstract}
 For relational structures $\A,\B$ of the same signature, the Promise Constraint Satisfaction
 Problem $\PCSP(\A,\B)$ asks whether a given input structure maps homomorphically to $\A$ or does not even map to $\B$.
 We are promised that the input satisfies exactly one of these two cases.

 If there exists a structure $\C$ with homomorphisms $\A\to\C\to\B$, then $\PCSP(\A,\B)$ reduces naturally to $\CSP(\C)$.
 To the best of our knowledge all known tractable PCSPs reduce to tractable CSPs in this way.
 However Barto~\cite{Ba:PMF} showed that some PCSPs over finite structures $\A,\B$ require solving CSPs over infinite $\C$.

 We show that even when such a reduction to some finite $\C$ is possible, this structure may become arbitrarily large.
 For every  integer $n>1$ and every prime $p$ we give $\A,\B$ of size $n$ with a single relation of arity $n^p$
 such that $\PCSP(\A,\B)$ reduces  via a chain of homomorphisms  $\A\to\C\to\B$ 
 to a tractable $\CSP$ over some $\C$ of size $p$ but not over any smaller structure.
 In a second family of examples, for every prime $p\geq 7$ we construct $\A,\B$ of size $p-1$ with a single
 ternary relation such that $\PCSP(\A,\B)$ reduces  via $\A\to\C\to\B$   to a tractable $\CSP$ over some $\C$ of size $p$ but not
 over any smaller structure.
 In contrast we show that if $\A,\B$ are graphs and $\PCSP(\A,\B)$ reduces to a tractable $\CSP(\C)$ for some finite digraph
 $\C$, then already $\A$ or $\B$ has a tractable $\CSP$.
 This extends results and answers a question of~\cite{DEMMNS}.
\end{abstract}
\maketitle

\section{Introduction}

 The \emph{Constraint Satisfaction Problem} ($\CSP$) for a fixed relational structure $\A$ can be formulated as the following
 homomorphism problem:
\dproblem{$\CSP(\A)$}{a relational structure $\X$}{yes, if there exists a homomorphism $\X\to\A$, \\ no, otherwise}

In~\cite{AGH:2ES14, AGH:2ES17} Austrin, Guruswami, and H\r{a}stad
introduced \emph{Promise Satisfaction Problems} ($\PCSP$) as approximations of
 $\CSP$. For relational structures $\A,\B$ of the same finite signature with a homomorphism $\A\to\B$, let
 \dproblem{$\PCSP(\A,\B)$}{a relational structure $\X$}{yes, if there exists a homomorphism $\X\to\A$, \\
   no, if there exists no homomorphism $\X\to\B$}
 Since we have a homomorphism $\A\to\B$ by assumption, the two alternatives 
 (1) $\X$ maps homomorphically to $\A$ or (2) $\X$ does not map homomorphically to $\B$
are mutually exclusive for any input $\X$. The \emph{promise} is that at least one
of the alternatives
 holds for $\X$.

 $\PCSP$s are motivated by open questions about the (in)approximability of SAT and graph coloring.
 The classical \emph{approximate graph coloring problem} for $r\leq s$ asks: given an $r$-colorable graph, find an $s$-coloring for
 it~\cite{GJ:CNOGC}.
 The decision version of this is to distinguish graphs that are $r$-colorable from those that are not even $s$-colorable.
 This is $\PCSP(\K_r,\K_s)$ for $\K_r,\K_s$ the complete graphs on $r,s$ vertices, respectively.
 It (and consequently the approximate graph coloring problem) has been conjectured to be $\compNP$-hard for all $3\leq r\leq s$.
 Even after more than 40 years of research this conjecture remains open in its full generality.
 
 Let $\A,\B,\C$ be relational structures of the same signature with homomorphisms $\A\to\C\to\B$. Then we say $\C$ is 
 \emph{sandwiched} by $\A$ and $\B$.
In this case $\PCSP(\A,\B)$ has a trivial reduction to $\CSP(\C)$ without changing the instance.
In general, the complexity of
the $\PCSP$ is unknown.
Still, to the best of our knowledge, all known tractable
(i.e., solvable in polynomial time)
PCSPs reduce to tractable CSPs in this way.
 
 As in~\cite{AB:FTPCSP} we call $\PCSP(\A,\B)$ \emph{finitely tractable} if there exists a finite $\C$ that is sandwiched by
 $\A$ and $\B$ such that $\CSP(\C)$ is tractable. 
 Barto~\cite{Ba:PMF} provided the first example of a PCSP
 over finite structures
 that is tractable but not finitely tractable.
 In~\cite{DEMMNS} the second author of this paper and his students gave the first example of a finitely tractable
 $\PCSP(\A,\B)$ for $\A,\B$ of size $2$ that does not reduce to a tractable $\CSP(\C)$ for any $\C$ of size $\leq 2$.

 We extend this result in Theorem~\ref{thm:example} as follows:
 For every prime $p$ and every integer $n>1$ we give $\A,\B$ of size $n$ such that $\PCSP(\A,\B)$ reduces to a
 tractable $\CSP$ over some $\C$ of size $p$ but not over any smaller structure.
 All these structures have a single relation with arity $n^p$.

 In our second family of examples in Theorem~\ref{thm:example2}
 we observe a similar behaviour even when restricting to structures with a single ternary relation:
 For every prime $p\geq 7$  we construct $\A,\B$ of size $p-1$ such that $\PCSP(\A,\B)$ reduces to a
 tractable $\CSP$ over some $\C$ of size $p$ but not over any smaller structure.

 In contrast we show that if finite undirected graphs $\A$ and $\B$ sandwich some finite directed graph $\C$ with $\CSP(\C)$
 tractable, then already  $\CSP(\A)$ or $\CSP(\B)$ is tractable (Corollary~\ref{cor:graphfintractable}).
 This answers Problem 1 of~\cite{DEMMNS}.

 More generally, if a finite smooth digraph $\A$ and a digraph $\B$ sandwich some finite digraph $\C$ with $\CSP(\C)$
 tractable, then already $\CSP(\A)$ or $\CSP(\B)$ is tractable (Corollary~\ref{cor:digraphfintractable}).
 It remains open whether every finitely tractable $\PCSP(\A,\B)$ for finite digraphs $\A,\B$ reduces to a tractable
 $\CSP(\C)$ for a digraph $\C$ of size $\leq\max(|A|,|B|)$.

 \section{Preliminaries}

We will only define a bare minimum of notions necessary to make sense of this article. For a more detailed introduction to promise constraint satisfaction, see~\cite{BBKO}. For more on digraph homomorphisms, see~\cite{HN:GH}.

 For $n\in\N$, write $\set{n} := \{0,\dots,n-1\}$ (note that $n\not\in \set{n}$).

 Unless indicated otherwise, arithmetical operations, like $+$ and $\cdot$, are considered over the integers
even if the input numbers come, say, from the set $\set{p}$. By $a\bmod p$ we mean the operation of taking remainder that returns a number in $\set{p}$. When writing formulas, $\bmod$ is evaluated after addition; in particular $\sum_i a_i\bmod p$ means $\left(\sum_i a_i\right)\bmod p$ where $\sum_i a_i$ is evaluated in $\mathbb Z$.

A \emph{relation} $R$ of arity $n$ on a set $A$ is a subset of $A^n$. A \emph{signature} is a (finite, in this article)
list of relation symbols, each of which is associated with an arity.
 A \emph{relational structure} of a given signature $\Gamma$,
 denoted by $\A$, consists of a universe $A$ together with a family of relations, one for each relation symbol from $\Gamma$.
 If $R$ is an $n$-ary relation symbol from $\Gamma$
and $\A$ is a structure over $\Gamma$ with universe $A$,
 then $R^\A$ is an $n$-ary relation on $A$. A relational structure is finite if its universe is finite.

 We call a relational structure $\A$ with universe $A$ \emph{affine} if there exists a binary operation $+$,
 a unary operation $-$ and a constant $0$ on $A$ such that $\mathbb{A} := (A,+,-,0)$ is an abelian group and every relation
 $R^\A$ of $\A$ is closed under $x-y+z$; equivalently, every, say, $n$-ary relation $R^\A$ is a coset of a subgroup
 of $\mathbb{A}^n$. Then $\CSP(\A)$ can be solved by linear algebra and is in the complexity class $\compP$.

 Let $n$ be a positive integer and $\A$ be a relational structure. By the \emph{$n$-th power of $\A$},
 denoted by $\A^n$, we will understand the relational structure with the universe $A^n$ (the set-theoretic $n$-th power)
 and with the same signature as $\A$. If $R^\A$ is an $m$-ary relation of $\A$, then $R^{\A^n}$ is the relation on
 $A^n$ that contains exactly those $m$-tuples $(u_1,u_2,\dots,u_m)\in (A^n)^m$ such that for all $i=1,2,\dots,n$
 the $i$-th projection of the $m$-tuple, $(\pi_i(u_1),\dots,\pi_i(u_m))$, lies in $R^\A$. Here $\pi_i$ is the mapping
 $A^n\to A$ that returns the $i$-th component of its input.

 Let $\A$ and $\B$ be two relational structures of the same signature $\Gamma$. A mapping $f\colon A\to B$ is a
 \emph{relational structure homomorphism} (or just homomorphism for short) if for each relation $R$ from $\Gamma$
 (denote its arity by $n$) and each $(a_1,\dots,a_n)\in R^\A$ we have $(f(a_1),\dots,f(a_n))\in R^\B$.
 An \emph{$n$-ary polymorphism} from $\A$ to $\B$ is a homomorphism $\A^n\to \B$.

 A finite relational structure $\A$ is a~\emph{core} if any homomorphism $\A\to\A$ is a bijective mapping.
 It is a well known fact that any finite relational structure has a unique (up to isomorphism) core.

Let $\Gamma$ be a (finite) signature and let $\A$ and $\B$ be two relational structures with signature $\Gamma$ such that there exists a homomorphism from $\A$ to $\B$.
\emph{The promise constraint satisfaction problem} with fixed target structures $\A$ and $\B$, denoted by $\PCSP(\A,\B)$, has as its instances finite relational structures $\X$ 
of signature $\Gamma$ (the relations of $\X$ are specified by listing all tuples in them). An instance is
\begin{itemize}
    \item a ``yes'' instance if there exists a  homomorphism $\X\to \A$, and
    \item a ``no'' instance if there does not exist any homomorphism $\X\to \B$.
\end{itemize}

It is easy to show that since $\A$ maps homomorphically to $\B$ no instance can be both a ``Yes'' and a ``No'' instance. However, there might exist instances that are neither ``Yes,'' nor ``No'' instances. On these inputs an algorithm solving $\PCSP(\A,\B)$ is allowed to do anything. A \emph{constraint satisfaction problem with target structure $\A$}, denoted by $\CSP(\A)$, is $\PCSP(\A,\A)$. 

As a reminder, $\compP$ is the class of all problems solvable in polynomial time and $\compNP$ is the class of all problems solvable in nondeterministic polynomial time.
Informally, we will call a (promise) constraint satisfaction problem tractable if there exists a polynomial time algorithm that correctly classifies each instance as a ``yes'' or ``no'' instance and we will assume that $\compP\neq\compNP$ since otherwise all presented complexity results are trivial.

If $\A$ is a relational structure and $\A'$ is the core of $\A$, then it is easy to verify that the sets of ``yes'' and
``no'' instances of $\CSP(\A)$ and $\CSP(\A')$ are the same. In particular, the computational complexity of $\CSP(\A)$ only
depends on the core of $\A$.

 Let $\A$ and $\A'$ be relational structures on the same universe but with possibly distinct signatures. We say that $\A'$ is
 \emph{pp-definable} from $\A$ if each relation on $\A'$ can be defined using a first order formula which only uses relations
 of $\A$, equality, existential quantifiers and conjunction. It is well-known and easy to see that if $\A'$ is pp-definable
 from $\A$, then $\CSP(\A')$ is reducible to $\CSP(\A)$.

 An important special type of relational structures are \emph{directed graphs (digraphs)};
 these are relational structures whose signature consists of one binary relation $E$, the edge relation.
 The elements of the universe of a digraph are usually called the vertices of the digraph. A symmetric (or
 undirected) \emph{graph} is a digraph whose edge relation is symmetric.
 A graph is \emph{bipartite} if its vertices can be partitioned into two sets $P$, $Q$ so that each edge connects
 a vertex from $P$ to some vertex from $Q$.
 A digraph $\G$ is \emph{smooth}
 if for each of its vertices $v$ there exist vertices $u$ and $w$ such that $(u,v)$ and $(v,w)$ both lie in $E^\G$.
 A \emph{directed cycle} of length $\ell$ is a digraph with vertices $[\ell]$ and edge set
 $\{(i,(i+1)\pmod \ell)\st i\in [\ell]\}$.

In the early 1990s, Hell and Nešetřil gave the following characterization of the complexity of CSP for symmetric graphs that we will later use:

\begin{thmC}[\protect{\cite[Theorem 1]{HN:CHC}}]\label{graph-dichotomy}
Let $\G$ be a finite symmetric graph. If $\G$ is bipartite, then $\CSP(\G)$ is in $\compP$. If $\G$ is not bipartite, then $\CSP(\G)$ is $\compNP$-complete. 
\end{thmC}

Barto, Kozik and Niven later generalized the above theorem to all smooth digraphs:
\begin{thmC}[\cite{BKN:CDHD}]\label{thm:digraph dichotomy}
Let $\G$ be a finite smooth digraph that is a core. If $\G$ is a disjoint union of directed cycles then $\CSP(\G)$ is in $\compP$. Otherwise, $\CSP(\G)$ is $\compNP$-complete.
\end{thmC}

 Our main tool for showing that some $\PCSP$ does not reduce to a tractable $\CSP(\C)$ for a structure $\C$ of
 specified size is the following result by Barto and Kozik. Recall that a function
$f\colon A^n\to B$ is \emph{cyclic} if for any $a_1,\dots,a_n\in A$ we have
\[
f(a_1,\dots,a_n)=f(a_2,a_3,\dots,a_n,a_1).
\]

\begin{thmC}[\protect{\cite[Theorem 4.1]{BK:ASCT}}]\label{thm:cyclic}
 Let $\C$ be finite relational structure, and let $p$ be a prime larger than the size of the universe of $\C$.
 If there exists no cyclic $p$-ary polymorphism $\C^p\to\C$, then $\CSP(\C)$ is $\compNP$-complete.
\end{thmC}

 If a structure $\C$ has a cyclic $p$-ary polymorphism $t$ and is sandwiched by $\A$ and $\B$ via homomorphisms
 $\A \xrightarrow{g} \C \xrightarrow{h} \B$, then the composition $f\colon A^p\to B$ defined by
 \[ f(x_1,\dots,x_p) := h(t(g(x_1),\dots,g(x_p))) \]
 is a cyclic $p$-ary polymorphism from $\A$ to $\B$.
 We will use the contrapositive of this statement together with Theorem~\ref{thm:cyclic} in the following form:
 if for some prime $p$ there is no cyclic $p$-ary polymorphism from $\A$ to $\B$, then $\PCSP(\A,\B)$ does not reduce
 to a tractable $\CSP(\C)$ for any $\C$ of size less than $p$.

\section{Big affine sandwiches}

 For every $n > 1$ and every prime $p$ we give structures $\A,\B$ of size $n$ with a single $n^p$-ary
 relation such that $\PCSP(\A,\B)$ reduces to a tractable $\CSP(\C)$ for some sandwiched $\C$ of size $p$ but
 not to a tractable $\CSP$ over any smaller sandwiched structure.
 For $n=2$ and $p=3$ such an example was given in~\cite{DEMMNS}.

\begin{thm}  \label{thm:example}
 For $n,p > 1$, 
 let $R$ be a relation symbol of arity $n^p$, and let 
 $\A = (\set{n}, R^\A)$, $\B = (\set{n}, R^\B),\C = (\set{p}, R^\C)$ be relational structures with
\begin{align*}
 R^\A & = \{ f\colon \set{n}^p\to \set{n} \st f \text{ is a projection} \}, \\
 R^\B & = \{ f\colon \set{n}^p\to \set{n} \st f \text{ is not cyclic} \}, \\
  R^\C & = \left\{ f\colon \set{n}^p\to \set{p},\ (x_1,\dots,x_p) \mapsto \sum_{i=1}^p a_ix_i \bmod p \st
     a_1,\dots,a_p\in\set{p}, \sum_{i=1}^p a_i \bmod p = 1 \right\}.
\end{align*}
 Then
\begin{enumerate}
\item \label{it:sandwich}
  $\C$ is affine (hence $\CSP(\C)$ is in $\compP$)
 and is sandwiched by $\A$ and $\B$ via the homomorphisms $\A \xrightarrow{g} \C \xrightarrow{h} \B$
 where $g\colon\set{n} \to \set{p},\ x \mapsto x \bmod p$, and $h\colon\set{p} \to \set{n},\ x \mapsto x \bmod n$.
\item \label{it:nosmaller} 
 If $p$ is prime and $\D$ is a structure with $|D|<p$ that is sandwiched by $\A$ and $\B$, then $\CSP(\D)$ is $\compNP$-complete.

\item \label{it:AB}
 $\PCSP(\A,\B)$ is in $\compP$. 
 If $p>2$, then $\CSP(\A)$ is $\compNP$-complete; else  it  is  in $\compP$.
 If $(n,p)\neq(2,2)$, then $\CSP(\B)$ is $\compNP$-complete; else  it is  in $\compP$.
\end{enumerate}
\end{thm}

 The main technical difficulty in proving Theorem~\ref{thm:example} is to show that $h\colon\C\to\B$ is a homomorphism.
 For this we first establish several properties of cyclic functions of a form that is slightly more general than of those
 in $h(R^\C)$ in the next lemma.

 For a rational number $q$ let $\lfloor q\rfloor$ denote the greatest integer less than or equal to $q$. 
 
\begin{lem} \label{lem:hf}
 For $n,p>1$ and $a_1,\dots,a_p\in\set{p}$, let
\[ f\colon \set{n}^p\to \set{n},\ (x_1,\dots,x_p) \mapsto \left(\sum_{i=1}^p a_ix_i \bmod p\right) \bmod n, \]
 be cyclic. Then
\begin{enumerate}
\item \label{claim:parity}
 $a_1,\dots,a_p$ are all congruent modulo $n$.
\item \label{claim:division}
 If $n$ does not divide $p$, then for any $x_1,\dots,x_p\in\set{n}$ we have
 \[  \left\lfloor \frac{\sum_{i=1}^p a_ix_i}{p} \right\rfloor = \left\lfloor \frac{\sum_{i=1}^p a_i\sum_{j=1}^p x_j}{p^2} \right\rfloor. \]
\item \label{claim:symmetric}
 $f$ is symmetric, i.e., invariant under all permutations of its inputs. 
\item \label{it:ai}
$\sum_{i=1}^p a_i \bmod p \neq 1$.
\end{enumerate}  
\end{lem}

\begin{proof}
 For~\eqref{claim:parity} note that $f(1,0\dots,0) = \dots = f(0,\dots,0,1)$ as $f$ is cyclic.
 Now observe that when the $1$ is at the $i$-th position,
 then $f(0,\dots,0,1,0,\dots,0)=a_i\bmod n$. Putting these two facts together gives us that all $a_i$'s are the same modulo $n$.

 For~\eqref{claim:division} assume that $n$ does not divide $p$.
 We show that the integer part of the quotient obtained by dividing $\sum_{i=1}^p a_ix_i$ by $p$ depends
 only on the \emph{weight} $w := \sum_{i=1}^p x_i$ of $(x_1,\dots,x_p)\in\set{n}^p$ by induction on $w$.

 The base case for $\sum_{i=1}^p x_i = 0$ is immediate.
 Next take $x_1,\dots,x_p\in\set{n}$ such that $\sum_{i=1}^p x_i = w > 0$. Without loss of generality assume
 $x_1\neq 0$.
 Let $\sigma$ be the cyclic permutation $\sigma := (1,\dots,p)$, and
 let $k := \left\lfloor \frac{\sum_{i=1}^p a_i (w-1)}{p^2} \right\rfloor$.  Let $y_1=x_1-1$ and $y_i=x_i$ for $i=2,3,\dots,p$. Since $\sum_{i=1}^p y_i=w-1$,
 the induction hypothesis yields
 \[
  \left\lfloor \frac{\sum_{i=1}^p a_i y_i}{p} \right\rfloor = \left\lfloor \frac{\sum_{i=1}^p a_i\sum_{j=1}^p y_j}{p^2} \right\rfloor=\left\lfloor \frac{\sum_{i=1}^p a_i(w-1)}{p^2} \right\rfloor=k.
 \]
 The total weight will not change if we permute the $y_i$'s, so for any $s\in [p]$ we also get
 $\left\lfloor \frac{\sum_{i=1}^p a_i y_{\sigma^s(i)}}{p} \right\rfloor = k$. In particular, $\sum_{i=1}^p a_i y_{\sigma^s(i)}-kp$ lies in $[p]$ for all  $s$.
 
Going back from $y_i$'s to $x_i$'s, we get
\[
\sum_{i=1}^p a_ix_{\sigma^s(i)} - a_{\sigma^{-s}(1)}- kp=\sum_{i=1}^p a_i y_{\sigma^s(i)}-kp \in\set{p}.
\]
 Given that $a_{\sigma^{-s}(1)}\in [p]$, we get
\begin{equation} \label{eq:in[2p-1]}
 \sum_{i=1}^p a_ix_{\sigma^s(i)} - kp \in\set{2p-1}
\end{equation}
for any $s\in\set{p}$.

Since $f$ is cyclic, $\sum_{i=1}^p a_ix_{\sigma^s(i)} \bmod p$ gives the same remainder modulo $n$ for all $s\in\set{p}$.
Furthermore, from~\eqref{eq:in[2p-1]} we see that for each $s$ we have
\[
\sum_{i=1}^p a_ix_{\sigma^s(i)}\bmod p=
\begin{cases}
\sum_{i=1}^p a_ix_{\sigma^s(i)}-kp, \text{ or}\\ \sum_{i=1}^p a_ix_{\sigma^s(i)}-(k+1)p.\\
\end{cases}
\]
We claim that exactly one of the two cases above holds for all $s\in\set{p}$.
Suppose for a contradiction that there are $q,r\in\set{p}$ such that
\begin{align*}
    \sum_{i=1}^p a_ix_{\sigma^q(i)}\bmod{p}&=\sum_{i=1}^p a_ix_{\sigma^q(i)}-kp\\
\sum_{i=1}^p a_ix_{\sigma^r(i)}\bmod{p}&=\sum_{i=1}^p a_ix_{\sigma^r(i)}-(k+1)p.
\end{align*}
 Since $f$ is cyclic, we get
\begin{equation}\label{eq:qrmodn}
\left(\sum_{i=1}^p a_ix_{\sigma^q(i)}-kp\right)\bmod n= 
\left(\sum_{i=1}^p a_ix_{\sigma^r(i)}-(k+1)p\right)\bmod n.
\end{equation}
 However, from~\eqref{claim:parity} we get that $\sum_{i=1}^p a_ix_{\sigma^s(i)}$ has the same remainder modulo $n$ for all
 $s\in\set{p}$.
 So~\eqref{eq:qrmodn} simplifies to $-kp\equiv -(k+1)p\pmod n$.
Thus $n$ divides $p$, which contradicts our assumption.

 We have obtained that one of the following two cases happens:
\begin{itemize}
 \item Either $\sum_{i=1}^p a_ix_{\sigma^{ s}(i)} - kp$  is in  $\set{p}$ for all $s\in\set{p}$, or
 \item $\sum_{i=1}^p a_ix_{\sigma^{ s}(i)} - kp$ is in $\{p,p+1,\dots, 2p-1\}$ for all $s\in\set{p}$.
 \end{itemize}
 In other words we have $k'\in\{k,k+1\}$ and $r_0\dots,r_{p-1}\in\set{p}$ such that
\begin{equation} \label{eq:k'}
 \sum_{i=1}^p a_ix_{\sigma^s(i)} = k'p+r_s 
\end{equation}
 for all $s\in\set{p}$. Dividing both sides of~\eqref{eq:k'} by $p$ we get
 $\left\lfloor\frac{\sum_{i=1}^p a_ix_{\sigma^s(i)}}{p}\right\rfloor=k'$. All that remains is to calculate $k'$.
  
 Summing~\eqref{eq:k'} over all $s\in \set{p}$, we get
\begin{equation}\label{eq:p2}
 \sum_{s=0}^{p-1}\sum_{i=1}^p a_ix_{\sigma^s(i)} = k'p^2+\sum_{s=0}^{p-1} r_s \leq k'p^2+p(p-1).
 \end{equation}
Counting in two ways, we obtain that the left hand side of~\eqref{eq:p2} is
\[ \sum_{s=0}^{p-1}\sum_{i=1}^p a_ix_{\sigma^s(i)} = \sum_{i=1}^p a_i \cdot \sum_{j=1}^p x_j. \]
Dividing by $p^2$ yields $k' =  \left\lfloor \frac{\sum_{i=1}^p a_i\sum_{j=1}^p x_j}{p^2} \right\rfloor$.
Thus in particular for $s=0$ we get
\[
\left\lfloor\frac{\sum_{i=1}^p a_ix_{\sigma^0(i)}}{p}\right\rfloor=\left\lfloor\frac{\sum_{i=1}^p a_ix_i}{p}\right\rfloor=\left\lfloor \frac{\sum_{i=1}^p a_i\sum_{j=1}^p x_j}{p^2} \right\rfloor
\]
and the induction step follows. Thus~\eqref{claim:division} is proved.

 Next we show~\eqref{claim:symmetric}. 
 If $n$ divides $p$, then $f$ simplifies to $f(x_1,\dots,x_p) = \sum_{i=1}^p a_ix_i \bmod n$, which is clearly symmetric
 by~\eqref{claim:parity}. So assume that $n$ does not divide $p$ in the following.
 Let $\pi$ be a permutation on $\{1,\dots,p\}$, and let $x_1,\dots,x_p\in\set{n}$.
 By~\eqref{claim:division} we have $k\in\N$ and $r,s\in\set{p}$ such that
 $\sum_{i=1}^p a_ix_{i}=kp+r$ and $\sum_{i=1}^p a_ix_{\pi(i)}=kp+s$. Since these sums give the same remainder modulo
 $n$ by~\eqref{claim:parity}, so do $r$ and $s$. 
 Hence $f(x_1,\dots,x_p) = f(x_{\pi(1)},\dots,x_{\pi(p)})$ for all permutations $\pi$ and~\eqref{claim:symmetric} is proved.

 Finally we prove~\eqref{it:ai}.
 Since $f$ is symmetric by~\eqref{claim:symmetric}, we may reorder its variables so that
\[ a_1\leq a_2 \leq \dots \leq a_p. \]
 Let $q\in\N$ be such that $p=2q$ or $p=2q+1$. For $b := a_{q+1}$, we have 
\[ r := \sum_{i=1}^{q} (a_i-b) \leq 0\quad \text{and}\quad s := \sum_{i=q+1}^p (a_i-b) \geq 0. \]
 Further
\begin{equation} \label{claim:rs}
  r+s = \sum_{i=1}^p a_i - pb \quad \text{and}\quad r \equiv s \equiv 0 \bmod n,
\end{equation}
 with the latter following from~\eqref{claim:parity}.

 If $n$ divides $p$, then~\eqref{claim:rs} implies that $n$ divides $\sum_{i=1}^p a_i$. In particular
 $\gcd(\sum_{i=1}^p a_i,p) \geq n > 1$ and~\eqref{it:ai} is proved in this case. So we assume that $n$ does not divide $p$
 for the rest of the proof.
 We apply the equality in~\eqref{claim:division} in the situation when either the first $q$ or the last $q$
 variables $x_i$ are $1$ and the rest of the $x_i$'s are $0$. In both cases $\sum_{j=1}^px_j=q$, so we get
 \[
  \left\lfloor \frac{\sum_{i=1}^{q} a_i}{p} \right\rfloor  =\left\lfloor \frac{\sum_{i=1}^p a_iq }{p^2} \right\rfloor =
  \left\lfloor \frac{\sum_{i=p-q+1}^{p} a_i}{p} \right\rfloor.  
 \]
 Replacing $\sum_{i=1}^{q} a_i$ by $r+qb$  and $\sum_{i=p-q+1}^{p} a_i$ by $s+qb$ yields
\[
\left\lfloor \frac{r+qb}{p} \right\rfloor  =
  \left\lfloor \frac{s+qb}{p} \right\rfloor.  
\]
Denote the quantity on the line above by $k$. Multiplying by $p$, we obtain the inequalities (recall that $r\leq s$)
\begin{equation} \label{eq:quotient2}
  kp \leq r+qb \leq s+qb \leq (k+1)p-1.
\end{equation} 
 Hence $0\leq s-r \leq p-1$. 
 Since $r\leq 0$, this yields $0\leq s\leq p-1$. Similarly $-p+1 \leq r\leq 0$. In particular $-p+1 \leq r+s \leq p-1$.

 Seeking a contradiction we suppose that $\sum_{i=1}^p a_i \equiv 1 \bmod p$. Then also $r+s\equiv 1 \bmod p$ by~\eqref{claim:rs},
 which leaves $r+s = -p+1$ or $r+s=1$.
 The latter case is impossible since $r$ and $s$ are both multiples of $n>1$ by~\eqref{claim:rs}.
 We are left with the case $r+s=-p+1$. Due to the inequalities $s\geq 0$ and $r\geq -p+1$, we must have $r=-p+1$ and $s=0$.
 We will now bring this case to a contradiction.
 
 Since $s-r=p-1$, the outer inequalities in~\eqref{eq:quotient2} must be equalities, that is,
 \[ r+qb=kp \quad \text{and} \quad s+qb=(k+1)p-1. \]
 We will again apply~\eqref{claim:division}; this time we will let either the first
 $q+1$ or the last $q+1$ variables $x_i$ be $1$ and the rest $0$. From~\eqref{claim:division} we obtain
 \[
 \left\lfloor \frac{r+qb+b}{p} \right\rfloor=
 \left\lfloor \frac{s+qb+a_{p-q}}{p} \right\rfloor.
 \]
 Plugging in $r+qb=kp$ and $s+qb=(k+1)p-1$, we get
 \begin{align*}
 \left\lfloor \frac{kp+b}{p} \right\rfloor&=
 \left\lfloor \frac{(k+1)p-1+a_{p-q}}{p} \right\rfloor, \\
 \left\lfloor k+\frac{b}{p} \right\rfloor&=
 \left\lfloor(k+1)+ \frac{a_{p-q}-1}{p} \right\rfloor.
 \end{align*}
 Since $b\leq p-1$ the left hand side evaluates to $k$. In order for the right hand side to also be $k$, we must have $a_{p-q}=0$.
 It follows that $a_1 = \dots = a_{q} \leq a_{p-q} = 0$ and consequently $r=-qb$. We distinguish two cases depending on the
 parity of $p$:
\begin{itemize}
\item   
 If $p = 2q+1$ is odd, then $0=a_{p-q}=a_{q+1}=b$ yields $r=0$, which contradicts $r=-p+1$. 
\item
  Else if $p = 2q$ is even, then $r=-qb$ and $r=-2q+1$ implies $q=b=1$ and $p=2$.
  In this situation we would have $a_1=0$ and $a_2=b=1$. Item~(\ref{claim:parity}) then gives us $0\equiv 1 \pmod n$,
  a contradiction with $n>1$.
\end{itemize} 
 Either case led to a contradiction. Thus~\eqref{it:ai} is proved.
\end{proof}  

 After this preparation we are now ready to prove our main result.

\begin{proof}[Proof of Theorem~\ref{thm:example}]
 For~\eqref{it:sandwich} note first that by definition $R^\C$ is the closure of $g(R^\A)$ under $x-y+z \bmod p$.
 Hence $\C$ is affine and $g\colon\A\to\C$ is a homomorphism. Next for any $f\in R^\C$, the function $h(f)$ is of the
 form
\[ h(f)\colon \set{n}^p\to \set{n},\ (x_1,\dots,x_p) \mapsto \left(\sum_{i=1}^p a_ix_i \bmod p\right) \bmod n, \]
for some $a_1,\dots,a_p\in\set{p}$ with $\sum_{i=1}^p a_i \bmod p = 1$.
Since $h(f)$ cannot by cyclic by part~\eqref{it:ai} of Lemma~\ref{lem:hf},
we have $h(f)\in R^\B$ and $h\colon\C\to\B$ is a homomorphism. 
 
 For~\eqref{it:nosmaller} assume that $p$ is prime.
 Let $\D$ be a relational structure whose universe is smaller than $p$ with homomorphisms
 $\A \xrightarrow{r} \D \xrightarrow{s} \B$. 
 Seeking a contradiction, suppose that $\D$ has some cyclic $p$-ary polymorphism $t\colon\D^p\to \D$.
 Consequently
\[ f\colon \A^p\to\B,\ (x_1,\dots,x_p) \mapsto s(t(r(x_1),\dots,r(x_p))), \]
 is a cyclic polymorphism from $\A$ to $\B$. However, applying $f$ to the graphs of projections
 $\pi_1$, $\pi_2,\dots,\pi_p \in R^\A$ gives us that $f\in R^\B$
 in contradiction to the definition of $R^\B$.

 Hence $\D$ has no cyclic $p$-ary polymorphisms. Thus $\CSP(\D)$ is $\compNP$-complete by Theorem~\ref{thm:cyclic}.

 For~\eqref{it:AB} note that $\PCSP(\A,\B)$ is in $\compP$ by~\eqref{it:sandwich}.
 If $n<p$ and $p$ is prime, then~\eqref{it:nosmaller} yields that $\CSP(\A)$ and $\CSP(\B)$ are $\compNP$-complete already.
 We still give brief self-contained arguments for the complexity of $\CSP(\A)$ and $\CSP(\B)$ in general.
 Let 
 \[
 E_p := \{(1,0,0,\dots,0),(0,1,0,\dots,0),\dots,(0,\dots,0,1)\}\subseteq [n]^p.
\]
First note that the \emph{projection} of $R^\A$ to $E_p$,
\[ \{ f|_{E_p} \st f\in R^{\A} \}, \]
is the set of restrictions of the $p$-ary projection maps to $E_p$. When $f\colon [n]^p\to [n]$ is a projection map,
then $f|_{E_p}$ can be identified with a member of $E_p$ in a natural way.
Hence $(\set{n}, E_p)$ is pp-definable from $\A$ and $\CSP(\set{2}, E_p)$ reduces to $\CSP(\A)$.
The former is known as $1$-in-$p$-SAT and  is $\compNP$-complete if $p>2$ by Schaefer's dichotomy for Boolean CSP~\cite{Sc:CSP}.
Thus $\CSP(\A)$ is $\compNP$-complete as well if $p>2$. 
 If $p=2$, then $R^\A$ contains only $2$ elements.
Hence $\A$ has a \emph{majority polymorphism} $m\colon \A^3\to\A$ satisfying $m(x,x,y) = m(x,y,x) = m(y,x,x) = x$ for all $x,y\in A$.
Thus $\CSP(\A)$ is tractable by~\cite{FV:CSMM}.

 Next consider the projection of $\{f\in R^\B \st f\text{ is constant on } [n]^p\setminus E_p \}$ to $E_p$.
 This is $N_p := [n]^p\setminus\{(a,\dots,a) \st a\in [n]\}$, the $p$-ary not-all-equal relation on $[n]$.
 Hence $\CSP([n], N_p)$ reduces to $\CSP(\B)$.
 If $n=2$, the former is Not-All-Equal $p$-SAT, which is known to be $\compNP$-complete for $p>2$ by Schaefer's dichotomy for Boolean CSP.
 For $n>2$, note that the not-equal relation $N_2 = \{ (a,b)\in [n]^2 \st (a,b\dots,b) \in N_p \}$ is primitively positively
 definable from $N_p$ and hence from $R^\B$. So $\CSP([n], N_2)$ reduces to $\CSP(\B)$. The former is graph $n$-coloring,
 which is $\compNP$-complete for $n>2$. It follows that $\CSP(\B)$ is $\compNP$-complete whenever $(n,p)\neq (2,2)$.

 If $p=n=2$, then $R^\B = \{ f\colon [2]^2\to [2] \st f(1,0) = f(0,1)+1 \bmod 2 \}$. Hence $\B$ is affine and $\CSP(\B)$
 in $\compP$. 
\end{proof}

 In our second class of examples,  let $p\geq 7$ be prime. We construct structures $\A,\B$ of size $p-1$  with a single ternary
 relation such that $\PCSP(\A,\B)$ reduces to a tractable $\CSP(\C)$ for some sandwiched $\C$ of size $p$ but
 not to a tractable $\CSP$ over any smaller sandwiched structure.

\begin{thm}  \label{thm:example2}
  For a prime $p \geq 7$,
 let $R$ be a relation symbol of arity $3$, and let 
 $\A = (\set{p-1}, R^\A)$, $\B = (\set{p-1}, R^\B),\C = (\set{p}, R^\C)$ be relational structures with
\begin{align*}
 R^\C & = \{ (x,y,z)\in\set{p}^3 \st x-2y+z \equiv 1 \bmod p \}, \\
 R^\A & = R^\C \cap \set{p-1}^3, \\
 R^\B & = h(R^\C) \text{ for } h\colon \set{p}\to\set{p-1},\ x\mapsto\begin{cases} x & \text{if } x\in\set{p-1}, \\
    1 & \text{if } x=p-1. \end{cases}
\end{align*}
 Then
\begin{enumerate}
\item \label{it:sandwich2}
 $\C$ is affine and sandwiched by $\A$ and $\B$ via the homomorphisms $\A \xrightarrow{\textrm{id}} \C \xrightarrow{h} \B$,
 and
\item \label{it:nosmaller2} 
 if $\D$ is a structure with $|D|<p$ that is sandwiched by $\A$ and $\B$, then $\CSP(\D)$ is $\compNP$-complete.
 In particular $\CSP(\A)$ and $\CSP(\B)$ are both $\compNP$-complete  but $\PCSP(\A,\B)$ is in $\compP$.
\end{enumerate}
\end{thm}

\begin{proof}
 Item~(\ref{it:sandwich2}) is immediate from the definitions. 
 Assertion~(\ref{it:nosmaller2}) will follow from Theorem~\ref{thm:cyclic} once we have proved that
 there exists no cyclic $p$-ary polymorphism $\A^p\to \B$. We will do that by showing several claims about $R^\A$ and $R^\B$. 

 To that end consider the linear system $M\cdot x=b$ over $\Z_p$ for
\begin{align*} M &:= \begin{pmatrix}
     1 & -2 & 1 & 0 && \dots && 0 \\
     0 & 1 & -2 & 1 && && \vdots \\
     \vdots &&\ddots&\ddots&\ddots&&& \vdots \\
    \vdots &&&\ddots&\ddots&\ddots&& \vdots \\
    \vdots &&&&\ddots&\ddots&\ddots& 0 \\
     0&&&&0 &1 & -2 & 1 \\     
     1 & 0 &  & && 0 & 1 & -2 \\
     -2 & 1 & 0 & & \dots&  & 0 & 1 \\
   \end{pmatrix} \in \Z_p^{p\times p},\quad
    b := \begin{pmatrix} 1\\1\\\vdots\\ \vdots\\\vdots\\\vdots\\1\\ 1\\\end{pmatrix} \in \Z_p^p.
\end{align*}
 To simplify notation we identify the elements of $\Z_p$ and of $[p]$ in the following.

\begin{clm} \label{nullspace}
 The nullspace of $M$ has dimension $2$ and is spanned by $a := (0, 1,  2, \dots, p-1)^T$ and $b$.
\end{clm}

\begin{proof}
 Clearly $a,b$ are linearly independent solutions for $M\cdot x = 0$. So the nullspace of $M$ has dimension at
 least $2$.
 On the other hand we see that the first $p-2$ rows of $M$ are linearly independent.
 So the nullspace of $M$ has dimension at most $2$.
\end{proof}

\begin{clm} \label{Mxb}
 $M\cdot x=b$ has a solution $x=(x_1,\dots,x_p)^T\in \Z_p^p$ with $x_{p-1}=x_p=0$.    
\end{clm}

\begin{proof}
 Let $r_1,\dots,r_p$ denote rows $1$ to $p$ of the augmented matrix $(M,b)$ of our linear system.
 Note that $r_{p-1},r_p,r_1,\dots,r_{p-2}$ are the rows of $(M^T,b)$.
 Claim~\ref{nullspace} and $a^T\cdot b = b^T\cdot b = 0$ yield
 \[ a^T \cdot (M^T,b) = (0,\dots,0) \text{ and } b^T \cdot (M^T,b) = (0,\dots,0). \] 
 Hence $(M^T,b)$ has rank at most $p-2$. Since $(M^T,b)$ and $(M,b)$ have the same set of rows (only permuted),
 it follows that $(M,b)$ also has rank at most $p-2$.

 We see that simply omitting the last two rows of $(M,b)$ yields a row echelon form for our linear system.
 Further variables $x_{p-1}$ and $x_p$ are free and can be set to $0$ for a solution of $M\cdot x =b$.
\end{proof}

\begin{clm}\label{claim:RA}
 There exist $u_1,\dots,u_p\in\set{p-1}$ such that each column of
\[ U := \begin{pmatrix}
u_1 & u_2 & \dots & u_p\\
u_2 & u_3 & \dots & u_1\\
u_3 & u_4 & \dots & u_2
\end{pmatrix} \]
is in $R^{\A}$.
\end{clm}

\begin{proof}
 By Claim~\ref{Mxb} $M\cdot x=b$ has a solution $x=(x_1,\dots,x_p)^T\in \Z_p^p$ such that the $p$ entries
 $\{x_1,\dots,x_p\}$ form a proper subset of $\Z_p$.
 Since $b$ is in the nullspace of $M$ by Claim~\ref{nullspace}, by adding an appropriate multiple of $b$ to $x$
 we can get a solution $(u_1,\dots,u_p)\in [p]^p$ such that $\{u_1,\dots,u_p\}$ does not contain $p-1$.
\end{proof}

\begin{clm}\label{claim:RB}
 $R^\B$ does not contain any constant tuple.
\end{clm}
  
\begin{proof}
 Seeking a contradiction, suppose we have $x,y,z\in\set{p}$ such that $x-2y+z\equiv 1 \bmod p$ and $h(x)=h(y)=h(z)$.
 Since $h$ restricted to $\set{p-1}$ is the identity, at least one of $x,y,z$ must be $p-1$. But then $h(x)=h(y)=h(z)=1$
 and $x,y,z\in \{1,p-1\}$. We distinguish the following cases:
\begin{itemize}
\item 
 If $x\neq z$, then $x-2y+z \equiv -2y \bmod p$. The latter is $1$ modulo $p$ only if $y=1$ and $p=3$.
\item
 If $x=z=1$, then $x-2y+z \equiv 2(1-y) \bmod p$. The latter is $1$ modulo $p$ only if $y=p-1$ and $p=3$.
\item
 If $x=z=p-1$, then $x-2y+z \equiv 2(p-1-y) \bmod p$. The latter is $1$ modulo $p$ only if $y=1$ and $p=5$.
\end{itemize}
 Each case yields a contradiction to our assumption that $p\geq 7$.
\end{proof}

 Finally let $\D$ be a structure such that $|D|< p$ with homomorphisms $\A \xrightarrow{r} \D \xrightarrow{s} \B$. 
 Seeking a contradiction, suppose that $\D$ has some cyclic $p$-ary polymorphism $t\colon\D^p\to \D$.
 Consequently
\[ f\colon \A^p\to\B,\ (x_1,\dots,x_p) \mapsto s(t(r(x_1),\dots,r(x_p))), \]
 is a cyclic polymorphism from $\A$ to $\B$. However, applying $f$ to the rows of the matrix $U$ (whose columns are 
 elements of $R^\A$) given by Claim~\ref{claim:RA} yields a constant tuple in $R^\B$.
 This contradicts Claim~\ref{claim:RB}.

 Hence $\D$ has no cyclic $p$-ary polymorphisms and $\CSP(\D)$ is $\compNP$-complete by Theorem~\ref{thm:cyclic}.
\end{proof}

\section{Symmetric structures}

 We show that there are no non-trivial (in the sense of the previous section) examples of finitely tractable $\PCSP$s over (undirected)
 graphs. For that we start with some observations on symmetric structures in general.

We call a $\Gamma$-structure $\A$ \emph{symmetric} if all its relations are invariant under all coordinate
 permutations, i.e., for every, say $k$-ary, $R\in\Gamma$ we have 
\[ \forall \sigma\in S_k\ \forall (x_1,\dots,x_k)\in R^\A: (x_{\sigma(1)},\dots,x_{\sigma(k)})\in R^\A. \]
Define the \emph {maximal symmetric subset $R^{\overline{\A}}$ of $R^\A$} as
\[R^{\overline{\A}}:=
\left\{(x_1,\dots,x_k)\in A^k \st \forall \sigma\in S_k\, (x_{\sigma(1)},\dots,x_{\sigma(k)})\in R^\A. 
\right\}\]
Since the condition for membership in $R^{\overline{A}}$ is a conjuction of $|S_k|$ many conditions, the \emph{symmetric part} $\overline{\A} := (A,\{ R^{\overline{\A}} \st R\in\Gamma\})$ of $\A$ is pp-definable from
 $\A$.
  
\begin{lem}  \label{lem:sym}
Let $g\colon\A \to \C$ be a homomorphism where $\A$ is symmetric and $\CSP(\C)$ is in $\compP$. 
Then $g$ maps $\A$ into $\overline{\C}$ (the symmetric part of $\C$) and $\CSP(\overline{\C})$ is in $\compP$. 
\end{lem}

\begin{proof}
 Since $\A$ is symmetric, so is its homomorphic image $g(\A)$ in $\C$.
 Hence $g(\A)$ is a substructure of $\overline{\C}$.
 Since $\overline{\C}$ is pp-definable from $\C$, $\CSP(\overline{\C})$ reduces to  $\CSP(\C)$ and
 so is in $\compP$.
\end{proof}
 
 Hence when searching for a  tractable $\C$ sandwiched by symmetric $\A$ and $\B$, we may restrict ourselves to
 searching for symmetric structures $\C$.

 As a consequence we can answer \cite[Problem 1]{DEMMNS} in the negative:
 Are there some finite graphs $\A,\B$ such that $\PCSP(\A,\B)$ is finitely tractable but only for a sandwiched digraph
 $\C$ that is strictly bigger than $\max(|A|,|B|)$?
 
\begin{cor} \label{cor:graph}
  If there exists a homomorphism from a finite (undirected)
  graph $\A$ into a finite digraph $\C$ without loops
  such that $\CSP(\C)$ is in $\compP$, then
  $\CSP(\A)$ is in $\compP$.  
\end{cor} 

\begin{proof}
 Assume $\compP \neq\compNP$ (clearly the assertion is trivial otherwise).  
 By Lemma~\ref{lem:sym} we may assume that $\C$ is an undirected graph (without loops).
 Since $\CSP(\C)$ is in $\compP$, $\C$ is bipartite by Theorem~\ref{graph-dichotomy}.
 From the homomorphism $\A\to\C$ we see that $\A$ is bipartite as well.
 Hence $\CSP(\A)$ is in $\compP$.
\end{proof}

 In particular finite tractability trivializes for $\PCSP$ on graphs.

\begin{cor} \label{cor:graphfintractable}
 If $\PCSP(\A,\B)$ for finite graphs $\A,\B$ is finitely tractable, then $\CSP(\A)$ or $\CSP(\B)$ is tractable.
\end{cor}

 This also settles a special instance of the following question which remains open in general:

\begin{prob}
 [Libor Barto, personal communication] Are there some finite symmetric $\A,\B$ such that $\PCSP(\A,\B)$ is finitely
 tractable but only for a sandwiched $\C$ that is strictly bigger than $\max(|A|,|B|)$?
\end{prob}

\section{Digraphs}

 In this section we move from symmetric graphs to directed graphs. The situation here is more complicated.
 In particular finite tractability does not trivialize for $\PCSP$ on digraphs like it does on symmetric graphs
 by Corollary~\ref{cor:graphfintractable}.
 The following example was communicated to us by Jakub Bul\'in: Let $\A$ be a finite tree with $\compNP$-complete CSP,
 let $\C,\B$ be the complete digraphs without loops on $2,3$ vertices, respectively.
 Clearly there exist homomorphisms $\A\to\C\to\B$ and $\CSP(\C)$ is tractable.
 Hence $\PCSP(\A,\B)$ is finitely tractable but $\A,\B$ have $\compNP$-complete CSPs.

 The following generalization of \cite[Problem 1]{DEMMNS} is still open:

\begin{prob}
 Are there some finite digraphs $\A,\B$ such that $\PCSP(\A,\B)$ is finitely
 tractable but only for a sandwiched $\C$ that is strictly bigger than $\max(|A|,|B|)$?
\end{prob}

 Brakensiek and Guruswami showed that every $\PCSP$ is polynomial time equivalent to $\PCSP(\A,\B)$ for
 some digraphs $\A,\B$~\cite[Theorem 6.10]{BG:PCSarxiv}. So our examples in Theorems~\ref{thm:example}
 and~\ref{thm:example2} can be translated to digraphs $\A'\to\C'\to\B'$
 but it is unclear whether the obtained $\C'$ is still the smallest sandwiched structure with tractable $\CSP$.
 
 Using the classification of finite smooth digraphs with tractable $\CSP$ in Theorem~\ref{thm:digraph dichotomy},
 we can get a partial answer that also generalizes Corollary~\ref{cor:graph}.
 
\begin{thm}\label{thm:digraph tractable}
 If there exists a homomorphism from a finite smooth digraph $\A$ into a finite digraph $\C$ without loops
 such that $\CSP(\C)$ is in $\compP$, then $\CSP(\A)$ is in $\compP$. 
\end{thm}   

\begin{proof}
 Assume $\compP \neq\compNP$ (clearly the assertion is trivial otherwise).  
 Let $\A$ be a finite smooth digraph, and let $\C$ be a minimal digraph without loops and with tractable
 $\CSP(\C)$ for which there exists a homomorphism $g\colon\A \to \C$.
 We proceed by a series of claims:

\begin{clm}
The digraph $\C$ is smooth.
\end{clm}
\begin{proof}

 Let $n=|C|$ and let $E^\C$ denote the edge relation of $\C$.
 By the pigeonhole principle the set $V$ of vertices of $\C$ that induces the maximal smooth subgraph of $\C$  is pp-definable as
 $v\in V$ if and only if
 \begin{align*}
 \exists x_1,\dots, x_n,y_1,\dots,y_n\st (x_1,x_2)\in E^\C \wedge (x_2,x_3)\in E^\C \wedge  \dots
 \wedge(x_{n-1},x_n)\in E^\C\\
 {}\wedge (x_n,v)\in E^\C \wedge (v,y_1)\in E^\C \wedge (y_1,y_2)\in E^\C \wedge \dots \wedge (y_{n-1},y_n)\in E^\C. 
 \end{align*}
 Then the induced subgraph $\C' := (V,E^\C|_{V\times V})$ is smooth and pp-definable from $\C$.
 Further $g(\A)$ is a subgraph of $\C'$ because $g(\A)$ is a smooth subgraph of $\C$.

 Since $\C'$ is pp-definable from $\C$, $\CSP(\C')$ reduces to  $\CSP(\C)$. Since the latter is tractable,
 so is the former. Hence $\C=\C'$ by the  minimality of $\C$ and the claim is proved.
\end{proof}

\begin{clm}
The digraph $\C$ is a core.
\end{clm}
\begin{proof}
Immediate from the minimality of $\C$.
\end{proof}

\begin{clm}\label{claim:cycles}
$\C$ is a disjoint union of directed cycles.
\end{clm}
\begin{proof}
  Since $\CSP(\C)$ is in $\compP$ and $\C$ a core, this follows from Theorem~\ref{thm:digraph dichotomy}.
\end{proof}
 \begin{clm}\label{claim:equal}
 The graph $g(\A)$ is equal to $\C$.
 \end{clm}
 
\begin{proof}
 Since $\A$ is smooth, $g(\A)$ is a smooth subgraph of $\C$. Now, $\C$ is a disjoint union of directed cycles by Claim~\ref{claim:cycles}. Since $g(\A)$ is a smooth subgraph of a disjoint union of directed cycles, we get that $g(\A)$
 itself is a disjoint union of some cycles of $\C$.  Claim~\ref{claim:equal} then follows from the minimality of $\C$.
 \end{proof}

\begin{clm} The digraph $\C$ is the core of $\A$.
\end{clm}
\begin{proof}
Let $\A'$ be a substructure of $\A$ that is isomorphic to the core of $\A$.
 Then $g(\A') = \C$ follows as in the proof of Claim~\ref{claim:equal}. Hence $\A'\cong\C$.
\end{proof}

Since $\CSP(\A)$ is equivalent to the $\CSP$ of its core $\C$ and $\C$ is a disjoint union of cycles (for whose
CSP there is a straightforward polynomial time algorithm), we get that $\CSP(\A)$ is in $\compP$,
concluding the proof.
\end{proof}

 As for undirected graphs we obtain the following consequence:

\begin{cor} \label{cor:digraphfintractable}
 If $\PCSP(\A,\B)$ for finite digraphs $\A,\B$ is finitely tractable and $\A$ is smooth,
 then $\CSP(\A)$ or $\CSP(\B)$ is tractable.
\end{cor}

\section{Acknowledgements}

We thank our anonymous reviewers for their insightful comments.

Alexandr Kazda was supported by the Charles University Research Centre program No. UNCE/SCI/022 and grant PRIMUS/21/SCI/014, by the INTER-EXCELLENCE project LTAUSA19070 M\v SMT Czech Republic, and by the Czech Science Foundation project GA \v CR 18-20123S. 
Dmitriy Zhuk worked on the paper 
at the National Research University Higher School of Economics and was 
supported by Russian Science Foundation (grant 20-11-20203).

\bibliographystyle{alphaurl}
\bibliography{pcsp.bib}

\newcommand{\etalchar}[1]{$^{#1}$}
\begin{thebibliography}{DSM{\etalchar{+}}21}

\bibitem[AB21]{AB:FTPCSP}
Kristina Asimi and Libor Barto.
\newblock Finitely tractable promise constraint satisfaction problems.
\newblock In {\em 46th {I}nternational {S}ymposium on {M}athematical
  {F}oundations of {C}omputer {S}cience}, volume 202 of {\em LIPIcs. Leibniz
  Int. Proc. Inform.}, pages Art. No. 11, 16. Schloss Dagstuhl. Leibniz-Zent.
  Inform., Wadern, 2021.

\bibitem[AGH17]{AGH:2ES17}
Per Austrin, Venkatesan Guruswami, and Johan H\r{a}stad.
\newblock $(2+\varepsilon)$-{S}at is {NP}-hard.
\newblock {\em SIAM Journal on Computing}, 46(5):1554--1573, 2017.
\newblock \href {https://doi.org/10.1137/15M1006507}
  {\path{doi:10.1137/15M1006507}}.

\bibitem[AHG14]{AGH:2ES14}
Per Austrin, Johan H\r{a}stad, and Venkatesan Guruswami.
\newblock (2 + epsilon)-{S}at is {NP}-hard.
\newblock In {\em 2014 IEEE 55th Annual Symposium on Foundations of Computer
  Science}, pages 1--10, 2014.
\newblock \href {https://doi.org/10.1109/FOCS.2014.9}
  {\path{doi:10.1109/FOCS.2014.9}}.

\bibitem[Bar19]{Ba:PMF}
Libor Barto.
\newblock Promises make finite (constraint satisfaction) problems infinitary.
\newblock In {\em 2019 34th Annual ACM/IEEE Symposium on Logic in Computer
  Science (LICS)}, pages 1--8, 2019.
\newblock \href {https://doi.org/10.1109/LICS.2019.8785671}
  {\path{doi:10.1109/LICS.2019.8785671}}.

\bibitem[BBKO21]{BBKO}
Libor Barto, Jakub Bul\'{\i}n, Andrei Krokhin, and Jakub Opr\v{s}al.
\newblock Algebraic approach to promise constraint satisfaction.
\newblock {\em Journal of the ACM}, 68(4):1--66, 2021.
\newblock \href {https://doi.org/10.1145/3457606} {\path{doi:10.1145/3457606}}.

\bibitem[BG17]{BG:PCSarxiv}
Joshua Brakensiek and Venkatesan Guruswami.
\newblock Promise {C}onstraint {S}atisfaction: Algebraic structure and a
  symmetric {B}oolean dichotomy, 2017.
\newblock Proceedings version appeared in SODA 2018, 39 pages. Available from
  \verb+https://arxiv.org/abs/1704.01937+.
\newblock URL: \url{https://arxiv.org/abs/1704.01937}.

\bibitem[BK12]{BK:ASCT}
Libor Barto and Marcin Kozik.
\newblock Absorbing subalgebras, cyclic terms, and the {C}onstraint
  {S}atisfaction {P}roblem.
\newblock {\em Logical Methods in Computer Science}, 8(1), 2012.
\newblock \href {https://doi.org/10.2168/LMCS-8(1:7)2012}
  {\path{doi:10.2168/LMCS-8(1:7)2012}}.

\bibitem[BKN09]{BKN:CDHD}
Libor Barto, Marcin Kozik, and Todd Niven.
\newblock The {CSP} dichotomy holds for digraphs with no sources and no sinks
  (a positive answer to a conjecture of {B}ang-{J}ensen and {H}ell).
\newblock {\em SIAM J. Comput.}, 38(5):1782--1802, 2008/09.
\newblock \href {https://doi.org/10.1137/070708093}
  {\path{doi:10.1137/070708093}}.

\bibitem[DSM{\etalchar{+}}21]{DEMMNS}
Guofeng Deng, Ezzeddine~El Sai, Trevor Manders, Peter Mayr, Poramate Nakkirt,
  and Athena Sparks.
\newblock Sandwiches for promise constraint satisfaction.
\newblock {\em Algebra Universalis}, 82(1):15, 2021.
\newblock \href {https://doi.org/10.1007/s00012-020-00702-5}
  {\path{doi:10.1007/s00012-020-00702-5}}.

\bibitem[FV99]{FV:CSMM}
Tom\'as Feder and Moshe~Y. Vardi.
\newblock The computational structure of monotone monadic {SNP} and constraint
  satisfaction: {A} study through {D}atalog and group theory.
\newblock {\em SIAM Journal on Computing}, 28(1):57--104 (electronic), 1999.
\newblock \href {https://doi.org/10.1137/S0097539794266766}
  {\path{doi:10.1137/S0097539794266766}}.

\bibitem[GJ76]{GJ:CNOGC}
Michael~Randolph Garey and David~S. Johnson.
\newblock The complexity of near-optimal graph coloring.
\newblock {\em Journal of the ACM}, 23(1):43--49, 1976.
\newblock \href {https://doi.org/10.1145/321921.321926}
  {\path{doi:10.1145/321921.321926}}.

\bibitem[HN90]{HN:CHC}
Pavol Hell and Jaroslav Ne\v{s}et\v{r}il.
\newblock On the complexity of {$H$}-coloring.
\newblock {\em Journal of Combinatorial Theory. Series B}, 48(1):92--110, 1990.
\newblock \href {https://doi.org/10.1016/0095-8956(90)90132-J}
  {\path{doi:10.1016/0095-8956(90)90132-J}}.

\bibitem[HN04]{HN:GH}
Pavol Hell and Jaroslav Ne\v{s}et\v{r}il.
\newblock {\em Graphs and homomorphisms}, volume~28 of {\em Oxford Lecture
  Series in Mathematics and its Applications}.
\newblock Oxford University Press, Oxford, 2004.
\newblock \href {https://doi.org/10.1093/acprof:oso/9780198528173.001.0001}
  {\path{doi:10.1093/acprof:oso/9780198528173.001.0001}}.

\bibitem[Sch78]{Sc:CSP}
Thomas~J. Schaefer.
\newblock The complexity of satisfiability problems.
\newblock In {\em Conference {R}ecord of the {T}enth {A}nnual {ACM} {S}ymposium
  on {T}heory of {C}omputing ({S}an {D}iego, {C}alif., 1978)}, pages 216--226.
  ACM, New York, 1978.
\newblock \href {https://doi.org/10.1145/800133.804350}
  {\path{doi:10.1145/800133.804350}}.

\end{thebibliography}
\end{document}